\documentclass[a4paper,11pt,reqno]{amsart}
\usepackage[body={16cm,24cm}]{geometry}
\usepackage[T1]{fontenc}
\usepackage{amssymb,amsmath}
\usepackage{mathtools}
\usepackage{enumerate}
\usepackage{graphicx}
\numberwithin{equation}{section}
\theoremstyle{plain}
\newtheorem{thm}{Theorem}[section]
\newtheorem{prop}[thm]{Proposition}
\newtheorem{cor}[thm]{Corollary}
\theoremstyle{definition}
\newtheorem*{rh-pb*}{RH problem}
\theoremstyle{remark}
\newtheorem{rem}[thm]{Remark}
\providecommand{\BS}[1]{\boldsymbol{#1}}
\providecommand{\D}[1]{\mathbb{#1}}
\newcommand{\dd}{\mathrm{d}}
\newcommand{\eul}{\mathrm{e}}
\newcommand{\ii}{\mathrm{i}}
\newcommand{\sign}{\mathrm{sign}}
\providecommand{\abs}[1]{\lvert#1\rvert}
\providecommand{\croch}[1]{\lbrack#1\rbrack}
\renewcommand{\Im}{\operatorname{Im}}
\newcommand{\ord}{\mathrm{O}}
\newcommand{\osmall}{\mathrm{o}}
\renewcommand{\Re}{\operatorname{Re}}
\DeclareMathOperator{\Res}{Res}
\begin{document}
\title[RHP for SP equation]{The short pulse equation\\by
a Riemann--Hilbert approach}
\author[A.~Boutet de Monvel]{Anne Boutet de Monvel$^{\ast}$}
\address{$^{\ast}$%
Institut de Math\'ematiques de Jussieu-PRG,
Universit\'e Paris Diderot,
Avenue de France,
B\^at.~Sophie Germain,
case 7012,
75205 Paris Cedex 13,
France}
\email{anne.boutet-de-monvel@imj-prg.fr}
\author[D.~Shepelsky]{Dmitry Shepelsky$^{\dagger}$}
\address{$^{\dagger}$%
Mathematical Division,
Institute for Low Temperature Physics,
47 Lenin Avenue,
61103 Kharkiv,
Ukraine}
\email{shepelsky@yahoo.com}
\author[L.~Zielinski]{Lech Zielinski$^{\ddagger}$}
\address{$^{\ddagger}$%
LMPA, Universit\'e du Littoral C\^ote d'Opale,
50 rue F. Buisson, 
CS 80699,
62228 Calais,
France}
\email{Lech.Zielinski@lmpa.univ-littoral.fr}
\date{\today}
\begin{abstract}
We develop a Riemann--Hilbert approach to the inverse scattering transform method for the short pulse (SP) equation 
\[
u_{xt}=u+\tfrac{1}{6}(u^3)_{xx}
\]
with zero boundary conditions (as $\abs{x}\to\infty$). This approach is directly applied to a Lax pair for the SP equation. It allows us to give a parametric representation of the solution to the Cauchy problem. This representation is then used for studying the long-time behavior of the solution as well as for retrieving the soliton solutions. Finally, the analysis of the long-time behavior allows us to formulate, in spectral terms, a sufficient condition for the wave breaking.
\end{abstract}
\maketitle
\section{Introduction}                   \label{sec:intro}

It is well-known that for describing the slow modulation of the amplitude of a weakly nonlinear wave packet in a moving medium, the nonlinear Schr\"odinger (NLS) equation is one of the universal nonlinear integrable models. It has been used with great success in nonlinear optics to describe the propagation of sufficiently broad pulses, or slowly varying wave trains whose spectra are narrowly localized around the carrier frequency. However, high-speed fiber-optic communication demands ultra-short pulses. With this respect, certain technological progress for creating them has been achieved; but it is important that in these conditions, the description of the evolution of these pulses lies beyond the usual approximations leading to the NLS equation. 

The short pulse (SP) equation 
\begin{equation}\label{spe}
u_{xt}=u+\frac{1}{6}(u^3)_{xx}
\end{equation}
was proposed by Sch\"afer and Wayne \cite{sw04,sw05} as an alternative (to the NLS equation) model for approximating the evolution of ultra-short intense infrared pulses in silica optics. It was shown in \cite{sw05} by numerical simulations that the SP equation can be successfully used for describing pulses with broad spectrum.

The short pulse equation is formally integrable: it is the compatibility condition for a pair of linear equations (the Lax pair), see \cite{ss05}:
\begin{subequations}\label{Lax}
\begin{align}\label{Lax-u}
&\Phi_x=U\Phi \\
\label{Lax-v}
&\Phi_t=V\Phi,
\end{align}
\end{subequations} 
where $U$ and $V$ are $2\times 2$ matrices dependent on the spectral parameter $\lambda$: 
\begin{subequations} \label{Lax-eqs}
\begin{align}\label{Lax-eqs-u}
&U=\begin{pmatrix}
\lambda \quad & \lambda u_x\\
\lambda u_x \quad& -\lambda \end{pmatrix},
\\
\label{Lax-eqs-v}
&V=
\begin{pmatrix} 
\frac{\lambda}{2} u^2+\frac{1}{4\lambda}&\frac{\lambda}{2} 
u^2 u_x-\frac12 u\\[5mm]
\frac{\lambda}{2} u^2 u_x+\frac{1}{2}u&-\frac{\lambda}{2} u^2-\frac{1}{4\lambda}
\end{pmatrix}.
\end{align}
\end{subequations}
 
The variants of application of the inverse scattering transform (IST) method to the SP equation, known in the literature, rely on establishing end exploiting the relationship between the SP equation and other integrable equations (like the sine-Gordon equation, see \cite{ss06}). But such relations turn out to be rather complicated and implicit, which, in particular, makes it difficult to apply them for studying initial value problems with general initial data. 

In the present paper we propose a direct approach to the problem of integration of the SP equation, which is based on applying the inverse scattering transform method, in the form of an associated Riemann--Hilbert (RH) problem. This means that the construction of this RH problem is made in terms of dedicated solutions of the Lax pair equations associated directly to the SP equation.

It is interesting to notice that the short pulse equation can be viewed as the short wave approximation to another integrable equation
\begin{equation}
m_t+\left((u^2-u_x^2)m\right)_x+u_x=0, \qquad m\coloneqq u-u_{xx},
\label{mCH}
\end{equation}
usually referred to as the ``modified Camassa--Holm equation'', and also known as the ``Fuchssteiner--Olver--Rosenau--Qiao'' equation
\cite{F96,FF81,OR96,Q06}. Indeed, introducing the new variables 
\[
x'=\frac{x}{\varepsilon},\quad t'=t\varepsilon,\quad u'=\frac{u}{\varepsilon^2},
\]
passing to the limit $\varepsilon\to 0$ and retaining the main terms reduce \eqref{mCH} to \eqref{spe}. With this respect we notice that the same short-wave limit applied to the Camassa--Holm (CH) equation
\begin{equation}    \label{CH-om}
u_t-u_{txx}+2 u_x+3uu_x=2u_xu_{xx}+uu_{xxx}
\end{equation}
leads to the so-called short wave (SW) equation
\begin{equation}  \label{mHS-om}
u_{txx}-2 u_x+2u_xu_{xx}+uu_{xxx}=0.
\end{equation}
The RH approach to the inverse scattering method for the CH equation and the SW equation were presented in \cite{BS06,BS08} and in \cite{BSZ11}, respectively. In what follows we will see that the development of this approach for the SP equation, on one hand, shares many common features with that for the SW equation, but on the other hand, has important differences.

Various aspects of the SP equation have been addressed in the literature, including the construction of solitary wave solutions \cite{G15,M07,ss06} and periodic solutions \cite{M08}. Well-posedness of the Cauchy problem has been studied in \cite{CR15,PS10,sw04}. Certain sufficient conditions for wave breaking have been found in \cite{LPS09}. An integrable hierarchy of equations associated with the SP equation is discussed in \cite{B05}. The bi-Hamiltonian structure of the SP equation is presented in \cite{B06}.

In this paper we present a RH problem formalism for the inverse scattering approach to the initial value problem for the SP equation: 
\begin{subequations}    \label{spe-ivp}
\begin{alignat}{3}           \label{spe-ic}
&u_{xt}=u+\frac16 (u^3)_{xx},&\quad&t>0,&\quad&-\infty<x<+\infty,\\
&u(x,0)=u_0(x),&&&&-\infty<x<+\infty \label{IC}.
\end{alignat}
\end{subequations}
We assume that $u_0(x)$ decays to $0$ sufficiently fast:
\[
u_0(x)\to 0,\qquad x\to\pm\infty,
\]
and we seek a solution $u(x,t)$ decaying to $0$ for all $t>0$:
\[
u(x,t)\to 0,\qquad x\to\pm\infty.
\]

In Section \ref{sec:2} we present appropriate Lax pairs associated with the SP equation, whose dedicated solutions are used in Section \ref{sec:3} for formulating a matrix Riemann--Hilbert problem suitable for solving the Cauchy problem \eqref{spe-ivp}. Then we give (Theorem \ref{thm:main}) a representation of the solution $u(x,t)$ of the problem \eqref{spe-ivp} in terms of the solution of this RH problem evaluated at a distinguished point of the complex plane of the spectral parameter. In Section \ref{sec:4} we discuss the construction of soliton solutions using the formalism of the RH problem. In Section \ref{sec:5} we study the long time asymptotics of the solution of the Cauchy problem \eqref{spe-ivp}. This study is then used in Section \ref{sec:5} to provide a sufficient condition for wave breaking of the solution of the Cauchy problem at a finite time.

\section{Lax pairs and eigenfunctions}\label{sec:2}

The RH formalism for integrable nonlinear equations makes use of the fact that it is possible to construct dedicated solutions of the linear equations from the associated Lax pair, which, being considered all together, are well-controlled, as functions of the spectral parameter, in the whole extended complex plane. These solutions are differently defined for different domains in the complex plane, and are related to each other at the boundaries between these domains. The latter fact is then interpreted as the ``analyticity defect'' for a (matrix-valued) function of the spectral parameter viewed as a function in the whole complex plane, and the inverse problem of the IST method for solving the Cauchy problem for the nonlinear equation in question is realized as a Riemann--Hilbert-type problem of reconstructing a piece-wise analytic function from the known ``analyticity defects'' of this function, in the form of jump conditions across certain contours supplemented by residue conditions (if any) at certain points in the complex plane of the spectral parameter.

An efficient approach to constructing such solution of the differential equations from the Lax pair is to pass to integral equations, whose solutions are particular solutions to the Lax pair equations. For this purpose, it is convenient to transform the Lax pair equations to a certain form, which is standard for establishing analytic properties of solutions near the singular points
(w.r.t.\ the spectral parameter) of the Lax pair equations.

Notice that the coefficient matrices $U$ and $V$ are traceless, which provides that the determinant of a matrix solution to \eqref{Lax} composed from two vector solutions is independent of $x$ and $t$.

In order to make the presentation close to that in the cases of other CH-type equations (see \cite{BS06,BS08,BSZ11,BS13}), it is convenient to introduce the spectral parameter $k\coloneqq\ii\lambda$.

Now notice that $U$ and $V$ have singularities (in the extended complex $k$-plane) at $k=0$ and at $k=\infty$. In order to control the behavior of solutions to \eqref{Lax} as functions of the spectral parameter $k$, we follow a strategy similar to that adopted for the CH equation \cite{BS06,BS08}. 

Namely, in order to control the large $k$ behavior of solutions of \eqref{Lax}, we will transform this Lax pair to the following form 
(see \cite{BC,BS06,BS08}):
\begin{subequations} \label{Lax-Q-form}
\begin{align}\label{Lax-Q-u}
&\hat\Phi_x+Q_x\hat\Phi=\hat U\hat\Phi,\\ 
\label{Lax-Q-v}
&\hat\Phi_t+Q_t\hat\Phi=\hat V\hat\Phi,
\end{align}
\end{subequations}
whose coefficients $Q(x,t,k)$, $\hat U(x,t,k)$, and $\hat V(x,t,k)$ have the following properties:
\begin{enumerate}[\rm(i)]
\item
$Q$ is diagonal and is unbounded as $k\to\infty$.
\item
$\hat U=\ord(1)$ and $\hat V=\ord(1)$ as $k\to\infty$.
\item
The diagonal parts of $\hat U$ and $\hat V$ decay as $k\to\infty$.
\item 
$\hat U\to 0$ and $\hat V\to 0$ as $x\to\pm\infty$.
\end{enumerate}

Since the coefficient matrix $U$ in \eqref{Lax-eqs-u} is the product of the spectral parameter and a matrix independent of the spectral parameter, in order to obtain $Q$, one has to diagonalize the latter matrix, i.e., to determine $P(x,t)$ such that
\[
PUP^{-1}=-Q.
\]
The freedom in determining such $P$ can be used in order to satisfy
item (iii) above, or, more precisely, to make the diagonal part of $\hat U$ identically $0$. 

Indeed, introducing 
\begin{equation}\label{w-q}
w:=u_x;\quad q\coloneqq\sqrt{1+w^2},
\end{equation}
setting
\begin{equation}\label{P}
P\coloneqq\sqrt{\frac{1+q}{2q}}
\begin{pmatrix}
1 & \frac{w}{1+q} \\ -\frac{w}{1+q} & 1
\end{pmatrix}
\end{equation}
so that $P^{-1}=\sqrt{\frac{1+q}{2q}}\left(\begin{smallmatrix}
	1 & -\frac{w}{1+q} \\ \frac{w}{1+q} & 1
\end{smallmatrix}\right)$, and introducing
\[
\hat\Phi\coloneqq P\Phi
\]
reduces \eqref{Lax-u} to \eqref{Lax-Q-u} with 
\begin{equation}\label{Q-x}
Q_x(x,t,k)=\ii k q(x,t) 
\begin{pmatrix}
1 & 0 \\ 0 & -1
\end{pmatrix}\equiv\ii kq(x,t)\sigma_3
\end{equation}
and 
\begin{equation}\label{U-hat}
\hat U=\hat U(x,t) 
= \frac{u_{xx}}{2q^2}
\begin{pmatrix}
0 & 1 \\ -1 & 0
\end{pmatrix}.
\end{equation}
Accordingly, the $t$-equation \eqref{Lax-v} of the Lax pair is transformed into 
\begin{equation}\label{t-V-check}
\hat \Phi_t=\check V \Phi,
\end{equation}
where
\begin{equation}\label{V-check}
\check V=-\frac{\ii k u^2 q}{2}\sigma_3 - \frac{1}{4\ii k q}\begin{pmatrix}
	1 & -w \\ -w & -1
\end{pmatrix}
+ \frac{u^2 u_{xx}}{4q^2}\begin{pmatrix}
	0 & 1 \\ -1 & 0
\end{pmatrix}.
\end{equation}
Noticing that $\check V\to -\frac{1}{4\ii k}\sigma_3$ as $x\to\pm\infty$, we write $\check V$ as 
\begin{align*}
\check V &= \left(-\frac{\ii k u^2 q}{2}-\frac{1}{4\ii k }\right)\sigma_3-\frac{1}{4\ii k q}\begin{pmatrix}
	1-q & -w \\ -w & -1+q
\end{pmatrix}
+ \frac{u^2 u_{xx}}{4q^2}
\begin{pmatrix}
0 & 1 \\ -1 & 0
\end{pmatrix}\\
&= \left(-\frac{\ii k u^2 q}{2}-\frac{1}{4\ii k }\right)\sigma_3
+ \hat V,
\end{align*}
where
\[
\hat V:=-\frac{1}{4\ii k q}\begin{pmatrix}
	1-q & -w \\ -w & -1+q
\end{pmatrix}
+ \frac{u^2 u_{xx}}{4q^2}\begin{pmatrix}
	0 & 1 \\ -1 & 0
\end{pmatrix},
\]
and thus \eqref{t-V-check} takes the form \eqref{Lax-Q-v} provided $Q$ is defined in such a way that 
\begin{equation}\label{Q-t}
Q_t=\left(\frac{\ii k u^2 q}{2}+\frac{1}{4\ii k }\right)\sigma_3.
\end{equation}
Now notice that \eqref{Q-x} and \eqref{Q-t} are compatible since the compatibility condition $Q_{xt}=Q_{tx}$ reads
\begin{equation}\label{cons-l}
q_t=\frac{1}{2}(u^2q)_x,
\end{equation}
which is the ``conservation law'' form of the short pulse equation \eqref{spe} (notice that the ``conservation law'' form of the SW equation \eqref{mHS-om} has a very similar form: $q_t=\frac{1}{2}(u q)_x$, but in that case, $q=\sqrt{1-u_{xx}}$). Thus $Q$ can be correctly defined by 
\begin{equation}\label{Q}
Q(x,t,k)=\left(\ii k\hat x(x,t)+\frac{t}{4\ii k}\right)\sigma_3,
\end{equation}
where
\begin{equation}\label{x-hat}
\hat x(x,t):=x-\int_x^\infty(q(y,t)-1)\dd y.
\end{equation}
Introduce 
\begin{equation}\label{zam}
\widetilde{\Phi}=\hat{\Phi}\eul^Q
\end{equation} 
and think about $\widetilde{\Phi}$ as a $2\times 2$ matrix. Then 
\eqref{Lax-Q-form} can be rewritten as 
\begin{subequations} \label{comsys}
\begin{align}
&\widetilde{\Phi}_x+\croch{Q_x,\widetilde{\Phi}}=\hat U\widetilde{\Phi},\\ 
&\widetilde{\Phi}_t+\croch{Q_t,\widetilde{\Phi}}=\hat V\widetilde{\Phi},
\end{align}
\end{subequations}
where $\croch{\,\cdot\,,\,\cdot\,}$ denotes the matrix commutator.

Now determine the particular (Jost) solutions $\widetilde{\Phi}_{\pm}(x,t)$ of \eqref{comsys} as the $2\times 2$ matrix-valued solutions of the associated Volterra integral equations:
\begin{equation}\label{inteq}
\widetilde{\Phi}_{\pm}(x,t,k)=I+\int_{\pm\infty}^{x}\eul^{Q(y,t,k)-Q(x,t,k)}\hat U(y,t,k)\widetilde{\Phi}_{\pm}(y,t,k)\eul^{Q(x,t,k)-Q(y,t,k)}\dd y,
\end{equation}
or, taking into account the definition of $Q$,
\begin{equation}\label{eq}
\begin{split}
\widetilde{\Phi}_{+}(x,t,k)=I-\int_x^{\infty}\eul^{\ii k\int_x^y{q(\xi,t)}\dd\xi\,\sigma_3}\hat U(y,t,k)\widetilde{\Phi}_{+}(y,t,k)\eul^{-\ii k\int_x^y{q(\xi,t)}\dd\xi\,\sigma_3}\dd y,\\
\widetilde{\Phi}_{-}(x,t,k)=I+\int_{-\infty}^x\eul^{-\ii k\int_y^x{q(\xi,t)}\dd\xi\,\sigma_3}\hat U(y,t,k)\widetilde{\Phi}_{-}(y,t,k)\eul^{\ii k\int_y^x{q(\xi,t)}\dd\xi\,\sigma_3}\dd y,
\end{split}
\end{equation}
where $I$ is the identity matrix.

We denote by $\mu^{(1)}$ and $\mu^{(2)}$ the columns of a $2\times 2$ matrix $\mu=\begin{pmatrix}\mu^{(1)}& \mu^{(2)}\end{pmatrix}$. Then it follows from \eqref{inteq} that for all $(x,t)$:
\begin{enumerate}[(i)]
\item
$\det\widetilde{\Phi}_\pm\equiv 1$.
\item 
$\widetilde{\Phi}_-^{(1)}$ and $\widetilde{\Phi}_+^{(2)}$ are analytic in $\{k\mid\Im k>0$ and continuous in $\{k\mid\Im k\geq 0,\, k\neq 0\}$.
\item 
$\widetilde{\Phi}_+^{(1)}$ and $\widetilde{\Phi}_-^{(2)}$ are analytic in $\{k\mid\Im k <0$ and continuous in $\{k\mid\Im k\leq 0,\, k\neq 0\}$.
\item
$\begin{pmatrix} 
\widetilde{\Phi}_-^{(1)} & \widetilde{\Phi}_+^{(2)}\end{pmatrix}\to I$ as $k\to\infty$ in $\{k\mid\Im k\geq 0\}$.
\item
$\begin{pmatrix} 
\widetilde{\Phi}_+^{(1)} & \widetilde{\Phi}_-^{(2)}\end{pmatrix}\to I$ as $k\to\infty$ in $\{k\mid\Im k\leq 0\}$.
\item
Symmetries:
\begin{equation}\label{sym-phi}
\overline{\widetilde{\Phi}_\pm(\,\cdot\,,\,\cdot\,,\bar k)}=\widetilde{\Phi}_\pm(\,\cdot\,,\,\cdot\,,-k)=\begin{pmatrix} 
0 & 1 \\ -1 & 0 \end{pmatrix}\widetilde{\Phi}_\pm(\,\cdot\,,\,\cdot\,,k)\begin{pmatrix} 0 & -1 \\ 1 & 0 \end{pmatrix}.
\end{equation}
Overline means complex conjugation for all $k$ for which the functions above are defined.
\end{enumerate}
The latter property is due to the symmetries of the coefficient matrix $\check U\coloneqq\hat U-\ii q\sigma_3$
\begin{equation}\label{sym}
\overline{\check U(\,\cdot\,,\,\cdot\,, \bar k)}=\check U(\,\cdot\,,\,\cdot\,, -k)=\begin{pmatrix} 0 & 1 \\ -1 & 0 \end{pmatrix}\check U(\,\cdot\,,\,\cdot\,, k)\begin{pmatrix}0&-1\\1&0\end{pmatrix}.
\end{equation}

\begin{rem}\label{rem:nls}
Introducing the new variable $\hat x$ as in \eqref{x-hat} and taking into account the bijectivity of the map $x\mapsto\hat x$ for any $t\geq 0$ (which is due to the fact that $q>0$), equation \eqref{Lax-Q-u} reduces to the (non-self-adjoint) Dirac equation
for $\Hat{\Hat\Phi}(\hat x,t,k)\coloneqq\hat\Phi(x(\hat x,t),t,k)$:
\begin{equation}
\Hat{\Hat\Phi}_{\hat x}+\ii k\sigma_3\Hat{\Hat\Phi}=\Hat{\Hat U}\Hat{\Hat\Phi},
\label{Lax-x-hat}
\end{equation}
where
\begin{equation}\label{U-2-hat}
\Hat{\Hat U}=\frac{u_{xx}}{2q^3}\begin{pmatrix}0&1\\-1&0 \end{pmatrix},
\end{equation}
which is the spatial equation from the Lax pair associated with the \emph{focusing} nonlinear Schr\"odinger (f\,NLS) equation; see, e.g., \cite{FT}. Therefore, the analytical properties of $\tilde\Phi_\pm$ stated above are the same as in the case of the f\,NLS equation. 
\end{rem}

\begin{rem}\label{rem:sw}
In the case of the SW equation \eqref{mHS-om}, the spatial equation from the Lax pair is also the Dirac equation (like \eqref{Lax-x-hat}), but with a self-adjoint potential $\Hat{\Hat U}$, see \cite{BSZ11}.
\end{rem}

The scattering matrix $s(k)$ (independent of $(x,t)$) is introduced by 
\begin{equation} \label{scat}
\widetilde{\Phi}_+(x,t,k)=\widetilde{\Phi}_-(x,t,k)\eul^{-Q(x,t,k)\sigma_3} s(k)\eul^{Q(x,t,k)\sigma_3},\qquad k\in\D{R}
\end{equation}
with $Q$ defined by \eqref{Q}, which, due to the symmetries \eqref{sym}, can be written in terms of two scalar spectral functions, $a(k)$ and $b(k)$: 
\begin{equation}\label{sym-s}
s(k)=\begin{pmatrix}\overline{a(k)}&b(k)\\-\overline{b(k)}&a(k) 
\end{pmatrix},\qquad k\in\D{R},
\end{equation}
such that $\overline{a(k)}=a(-k)$ and $\overline{b(k)}=b(-k)$. In view of Remark \ref{rem:nls}, the spectral functions have properties similar to those in the case of the focusing NLS equation \cite{FT}:
\begin{enumerate}[a)]
\item
$a(k)$ and $b(k)$ are determined by $u(x,0)$ through the solutions $\widetilde{\Phi}_\pm(x,0)$ of equations \eqref{eq}, where $\hat U=\hat U(x,0)$ is defined by \eqref{U-hat} with $u$ replaced by $u_0(x)$ (and similarly for $q$).
\item 
$a(k)$ is analytic in $\{k\mid\Im k>0\}$ and continuous in $\{k\mid\Im k\geq 0\}$; moreover, $a(k)\to 1$ as $k\to\infty$.
\item 
$b(k)$ is continuous for $k\in\D{R}$, and $b(k)\to 0$ as $\abs{k}\to\infty$.
\item
$\abs{a(k)}^2+\abs{b(k)}^2=1$ for $k\in\D{R}$.
\item 
Let $\{k_j\}_1^N$ be the set of zeros of $a(k)$. We make the \emph{genericity assumption} that these zeros are finite in number, simple, and no zero is real. Then $\widetilde{\Phi}^{(1)}_{-}(x,t,k_j)$ and $\widetilde{\Phi}^{(2)}_{+}(x,t,k_j)$ are linearly dependent; moreover,
\begin{equation}
\widetilde{\Phi}^{(1)}_{-}(x,t,k_j)=\eul^{2\ii(k_j\hat{x}-\frac{t}{4k_j})}\widetilde{\Phi}^{(2)}_{+}(x,t,k_j)\chi_j
\label{res-phi}
\end{equation}
with some constants $\chi_j$.
\end{enumerate}

\section{The Riemann--Hilbert problem}\label{sec:3}
\subsection{A RH problem constructed from dedicated eigenfunctions}

The analytic properties of $\widetilde{\Phi}_\pm$ stated above allow rewriting the scattering relation \eqref{scat} as a jump relation for a piece-wise meromorphic (w.r.t.\ $k$), $2\times 2$-valued function (depending on $x$ and $t$ as parameters). Indeed, define $M(x,t,k)$ by
\begin{equation}\label{M}
M(x,t,k)=
\begin{cases}
\begin{pmatrix}\frac{\widetilde{\Phi}_-^{(1)}(x,t,k)}{a(k)} & \widetilde{\Phi}_+^{(2)}(x,t,k)\end{pmatrix},
	&\  \Im k>0, \\
\begin{pmatrix}\widetilde{\Phi}_+^{(1)}(x,t,k) & \frac{\widetilde{\Phi}_-^{(2)}(x,t,k)}{\overline{a(k)}}\end{pmatrix},
	&\ \Im  k<0.
\end{cases}
\end{equation} 
Define 
\begin{equation}  \label{refl}
r(k):=-\frac{\overline{b(k)}}{a(k)}\quad\text{for }k\in\D{R}.
\end{equation}
Then the limiting values $M_\pm(x,t,k)$, $k\in\D{R}$ of $M$ as $k$ is approached from the domains $\pm\Im k>0$ are related as follows:
\begin{equation}\label{RH-x}
M_+(x,t,k)=M_-(x,t,k)\eul^{-Q(x,t,k)\sigma_3}J_0(k)\eul^{Q(x,t,k)\sigma_3},\quad k\in\D{R},
\end{equation}
where 
\begin{equation}\label{J0}
J_0(k)=\begin{pmatrix}
1+\abs{r(k)}^2  & \overline{r(k)} \\ r(k)  & 1
\end{pmatrix}.
\end{equation}
Taking into account the properties of $\widetilde{\Phi}_\pm$ and $s(k)$, $M(x,t,k)$ satisfies the following properties:
\begin{enumerate}[(i)]
\item 
$\det M\equiv 1$.
\item
Normalization: $M(\,\cdot\,,\,\cdot\,,k)\to I$ as $k\to\infty$.
\item
Symmetries:
\begin{equation}
\overline{M(\,\cdot\,,\,\cdot\,,\bar k)}=M(\,\cdot\,,\,\cdot\,,-k)=\begin{pmatrix} 0 & 1 \\ -1 & 0 \end{pmatrix}M(\,\cdot\,,\,\cdot\,,k)\begin{pmatrix} 0 & -1 \\ 1 & 0 \end{pmatrix}.
\label{M-sym}
\end{equation}
\item 
$M^{(1)}$ has poles at the zeros $k_j$ of $a(k)$ (in $\{k\mid\Im k >0\}$), whereas $M^{(2)}$ has poles at the conjugates $\overline{k_j}$ (in $\{k\mid\Im k <0\}$), $j=1,2,\dots,N$, and the following residue conditions are satisfied:
\begin{equation}\label{res-M}
\begin{split}
\Res_{k=k_j}M^{(1)}(x,t,k)&=\ii\gamma_j\eul^{2\ii\bigl(k_j\hat{x(x,t)}-\frac{t}{4k_j}\bigr)}M^{(2)}(x,t,k_j),\\
\Res_{k=\bar k_j}M^{(2)}(x,t,k)&=\ii\bar\gamma_j 
\eul^{2\ii\bigl(\bar k_j\hat{x}(x,t)-\frac{t}{4\bar k_j}\bigr)}M^{(1)}(x,t,\bar k_j)
\end{split}
\end{equation}
with some constants $\gamma_j$.
\end{enumerate}

The idea of the Riemann--Hilbert problem approach in the inverse scattering method consists in considering the jump relation \eqref{RH-x} complemented by the normalization condition $M\to I$ as $k\to\infty$, and by the residue conditions \eqref{res-M}, as the factorization problem of finding $M(x,t,k)$ (and, consequently, $u(x,t)$) from the jump matrix in \eqref{RH-x} and the residue conditions \eqref{res-M} at the singularities of $M$. As in the case of any Camassa--Holm-type equation, when realizing this idea, one faces the problem that the determination of the jump matrix, which is $\eul^{-Q}J_0(k)\eul^{Q}$, involves not only objects uniquely determined by the initial data $u(x,0)$ (the functions $a(k)$ and $b(k)$ involved in $J_0(k)$ and the constants involved in the residue conditions), but also $Q=Q(x,t,k)$, which is obviously not determined by $u(x,0)$ (it involves $u(x,t)$ for $t\geq 0$). This problem can be resolved by considering a RH problem depending, instead of $(x,t)$, on the parameters $\hat x$ and $t$. Then the jump matrix and the residue conditions become explicit in the new variables $\hat x$ and $t$. Indeed, introducing
\[
\hat M(\hat x,t,k)\coloneqq M(x(\hat x,t),t,k),
\]
the jump condition takes the form
\begin{equation}\label{jump-y}
\hat M_+(\hat x,t,k)=\hat M_-(\hat x,t,k)J(\hat x,t,k),\qquad k\in\D{R},
\end{equation}
where
\begin{equation}\label{J-J0}
J(\hat x,t,k)\coloneqq\eul^{-\hat Q(\hat x,t,k)}J_0(k)\eul^{\hat Q(\hat x,t,k)}
\end{equation}
with $J_0(k)$ as in \eqref{J0} and 
\begin{equation}\label{Q-hat}
\hat Q(\hat x,t,k)\coloneqq\left(\ii k\hat x+\frac{t}{4\ii k}\right)\sigma_3,
\end{equation}
Accordingly, the residue conditions \eqref{res-M} take the form
\begin{equation}\label{res-M-hat}
\begin{split}
\Res_{k=k_j}\hat M^{(1)}(\hat x,t,k)&=\ii\gamma_j 
\eul^{2\ii \bigl(k_j\hat{x}-\frac{t}{4k_j}\bigr)}\hat M^{(2)}(\hat x,t,k_j),\\
\Res_{k=\bar k_j}\hat M^{(2)}(\hat x,t,k)&=\ii\bar\gamma_j 
\eul^{-2\ii\bigl(\bar k_j\hat{x}-\frac{t}{4\bar k_j}\bigr)}\hat M^{(1)}(\hat x,t,\bar k_j).
\end{split}
\end{equation}

Recall that the jump and residue conditions for $\hat M(\hat x,t,k)$ were obtained above assuming that there exists a solution $u(x,t)$ of the SP equation decaying to $0$ as $x\to\pm\infty$ for any fixed $t>0$. On the other hand, the conditions \eqref{jump-y}--\eqref{res-M-hat} can be viewed as a factorization problem of Riemann--Hilbert (RH) type:

\begin{rh-pb*}
Given $\{r(k),\,k\in\D{R};\{k_j,\gamma_j\}_{1}^N\}$, find a piece-wise (w.r.t.\ $\D{R}$) meromorphic function $\hat M(\hat x,t,k)$ satisfying the conditions \eqref{jump-y}--\eqref{res-M-hat} complemented by the normalization condition 
\begin{equation}
\hat M(\hat x,t,k)\to I\quad\text{ as}\ k\to\infty.
\label{M-norm}
\end{equation}
\end{rh-pb*}

\begin{rem}[symmetries]\label{rem:sym}
Since the jump matrix $J$ satisfies the symmetry conditions described in \eqref{M-sym}, it follows from the uniqueness of the solution $\hat M$ of the RH problem that this solution satisfies \eqref{M-sym} as well.
\end{rem}

\begin{rem}[unique solvability]\label{rem:nls2}
The structure of the jump matrix and the residue conditions are the same as in the case of the focusing NLS equation (only the dependence on $\hat x$ and $t$, which are just the \emph{parameters} for the RH problem, is different), which implies that for all $\hat x$ and $t$, there exists a unique solution of the RH problem \eqref{jump-y}--\eqref{M-norm} provided that $r(k)$ can be represented as $r(k)=\int_{-\infty}^\infty\hat r(s)\eul^{\ii k s}\dd s$ with some $\hat r(s)\in L^1(-\infty,\infty)$, see \cite{FT}.
\end{rem}

\subsection{Recovering the solution of the Cauchy problem from the associated RH problem}

Now our goal is to show that $u(x,t)$ can be recovered in terms of $\hat M(\hat x,t,k)$, which is considered as the solution of the Riemann--Hilbert problem \eqref{jump-y}--\eqref{M-norm} (notice that the data for this problem are uniquely determined by $u(x,0)$), evaluated at $k=0$. Indeed this value of $k$ is specific for \eqref{Lax-u}, since the coefficient matrix $U$ in \eqref{Lax-u} vanish identically at $k=0$.

In order to have a good control of the behavior of $\hat M(\hat x,t,k)$ as $k\to 0$, it is convenient to rewrite the Lax pair \eqref{Lax} in the form
\begin{subequations}\label{Lax-0}
\begin{align}\label{Lax-0-u}
&\Phi_x+\ii k \sigma_3\Phi=U_0 \Phi \\
\label{Lax-0-v}
&\Phi_t+\frac{1}{4\ii k} \sigma_3\Phi=V_0 \Phi,
\end{align}
\end{subequations} 
where 
\begin{subequations} \label{Lax-eqs-0}
\begin{align}\label{Lax-eqs-0-u}
&U_0=-\ii kw\begin{pmatrix}
0 & 1\\
1 & 0\end{pmatrix},
\\
\label{Lax-eqs-0-v}
&V_0= -\frac{\ii ku^2}{2}\begin{pmatrix}
1 & w\\
w & -1\end{pmatrix}+\frac{u}{2}
\begin{pmatrix} 0 & -1 \\ 1 & 0
\end{pmatrix}.
\end{align}
\end{subequations}
Notice that $U_0\to 0$ and $V_0\to 0$ as $\abs{x}\to\infty$. Besides, it is important that $U_0(x,t,0)\equiv 0$.

Introduce 
\begin{equation}
Q_0(x,t,k)\coloneqq\left(\ii kx+\frac{t}{4\ii k}\right)\sigma_3
\label{Q-0}
\end{equation}
and 
\begin{equation}\label{zam0}
\widetilde{\Phi}_0={\Phi}\eul^{Q_0}.
\end{equation} 
Then \eqref{Lax-0} can be rewritten as 
\begin{subequations} \label{comsys-0}
\begin{align}
&\widetilde{\Phi}_{0x}+\croch{Q_{0x},\widetilde{\Phi}_0}= U_0\widetilde{\Phi}_0,\\ 
&\widetilde{\Phi}_{0t}+\croch{Q_{0t},\widetilde{\Phi}_0}= V_0\widetilde{\Phi}_0.
\end{align}
\end{subequations}
The Jost solutions $\widetilde{\Phi}_{0\pm}(x,t,k)$ of \eqref{comsys-0} are determined, similarly to above, as the solutions of associated Volterra integral equations:
\begin{equation}\label{inteq-0}
\widetilde{\Phi}_{0\pm}(x,t,k)=I+\int_{\pm\infty}^x\eul^{\ii k(y-x)} U_0(y,t,k)\widetilde{\Phi}_{0\pm}(y,t,k)\eul^{\ii k (x-y)}\dd y.
\end{equation}

Now, since $U_0(x,t,0)\equiv 0$, we have the following important property:
\begin{equation}
\widetilde{\Phi}_{0\pm}\left(x,t,0\right)\equiv I
\label{struct}
\end{equation}
for all $x$ and $t$. Moreover, directly using \eqref{inteq-0} one obtains

\begin{prop}\label{prop:expand}
As $k\to 0$,
\[
\widetilde{\Phi}_{0\pm}(x,t,k)=I-\ii ku(x,t)
\begin{pmatrix}
0 & 1 \\ 1 & 0
\end{pmatrix}+\ord(k^2).
\]
\end{prop}

Further, we notice that $\widetilde{\Phi}_{\pm}$ and $\widetilde{\Phi}_{0\pm}$, being related to the same system of equations \eqref{Lax}, are related as 
\begin{equation}\label{zvazok}
\widetilde{\Phi}_{\pm}(x,t,k)=P(x,t)\widetilde{\Phi}_{0\pm}(x,t,k)\eul^{-Q_0(x,t,k)}C_{\pm}(k)\eul^{Q(x,t,k)},
\end{equation}
where $C_{\pm}(k)$ are some matrices independent of  $x$ and $t$. Passing to the limits $x\to\pm\infty$ determines $C_{\pm}(k)$:
\[
C_+(k)=I,\qquad C_-(k)=\eul^{\ii k \alpha}\sigma_3,
\]
where $\alpha=\int_{-\infty}^\infty(q(y,t)-1)\dd y$; notice that in view of \eqref{cons-l}, $\alpha$ is constant (does not depend on $t$).

Combining Proposition \ref{prop:expand} with \eqref{zvazok} one gets
\begin{equation}\label{phi-exp}
\begin{split}
\widetilde{\Phi}_+(x,t,k)&=P(x,t)\left(
I-\ii k\left(u(x,t)\sigma_1+\int_x^\infty (q(y,t)-1)\dd y\  \sigma_3\right)+\ord(k^2)\right),\\
\widetilde{\Phi}_-(x,t,k)&=P(x,t)\left(
I-\ii k\left(u(x,t)\sigma_1 -\int^x_{-\infty} (q(y,t)-1)\dd y\  \sigma_3\right)+\ord(k^2)\right)
\end{split}
\end{equation} 
as $k\to 0$. Using these expansions in \eqref{scat} we expand $s(k)$ at $k=0$, then by \eqref{sym-s} we obtain
\begin{equation}
a(k)=1+\ii k\alpha+\ord(k^2),\qquad b(k)=\ord(k^2),\quad k\to 0.
\label{ab-exp}
\end{equation}
Finally, substituting \eqref{phi-exp} and \eqref{ab-exp} into 
\eqref{M} gives
\begin{equation}\label{M-k-0}
M(x,t,k)=P(x,t)\left(I-\ii k\left(u(x,t)\sigma_1 +\int_x^\infty (q(y,t)-1)\dd y\ \sigma_3\right)+\ord(k^2)\right),\quad k\to 0,
\end{equation}
which, in view of \eqref{x-hat}, reads
\begin{equation}\label{M-k-0-hat}
M(x,t,k)=P(x,t)\left(I-\ii k\begin{pmatrix}
x-\hat x & u \\ u & \hat x - x
\end{pmatrix}+ \ord(k^2)\right),\quad k\to 0.
\end{equation}

This relation \eqref{M-k-0-hat} leads to the following result.

\begin{thm}[representation result]\label{thm:main}
Assume that the Cauchy problem \eqref{spe-ivp} for the SP equation has a solution $u(x,t)$s. Let $\{r(k),\,k\in\D{R};\,\{k_j,\gamma_j\}_{1}^N\}$ be the spectral data determined by $u_0(x)$, and let $\hat M(\hat x,t,k)$ be the solution of the associated RH problem \eqref{jump-y}--\eqref{M-norm}. Then, evaluating $\hat M$ as $k\to 0$, we get a parametric representation for the solution $u(x,t)$ of the Cauchy problem \eqref{spe-ivp}:
\begin{subequations}\label{eq:main}
\begin{equation}\label{eq:rec}
u(x,t)=\hat u(\hat x(x,t),t)
\end{equation}
with
\begin{align}\label{xofy}
x(\hat x,t)&=\hat x+f_1(\hat x,t),\\
\hat u(\hat x,t)&=f_2(\hat x,t),
\label{uofx}
\end{align}
where 
\begin{equation}
\begin{pmatrix}
f_1 & f_2 \\ f_2 & -f_1
\end{pmatrix}(\hat x,t)\coloneqq\lim_{k\to 0}\frac{\ii}{k}\left(\hat M^{-1}(\hat x,t,0)\hat M(\hat x,t,k)-I\right).
\label{rep}
\end{equation}
\end{subequations}
\end{thm}

\begin{rem}
The representation result of Theorem \ref{thm:main} can be interpreted in two ways:
\begin{enumerate}[i)]
\item 
If there is a global (in time) classical solution $u(x,t)$ of the Cauchy problem \eqref{spe-ivp}, then \eqref{eq:main} gives a parametric representation of this solution for all $t$.
\item
If wave breaking occurs at a finite time, then the bijectivity of the map $\hat x\mapsto x$ described by \eqref{xofy} is broken for certain values of $t$. Then \eqref{xofy} and \eqref{uofx} present a continuation of the solution of the Cauchy problem \eqref{spe-ivp} after the wave breaking. Particularly, if the bijectivity of the map $\hat x\mapsto x$ is restored for all times $t$ greater than a certain $T$, \eqref{xofy} and \eqref{uofx} present a solution that, after undergoing a sequence of wave breakings, retrieves the form of a classical solution to the SP equation \eqref{spe-ic}.
\end{enumerate}
\end{rem}

\begin{rem}
We emphasize that the wave breaking mechanism for the SP equation is exclusively related to the break of bijectivity in \eqref{xofy} while the solution on the variables $(\hat x,t)$ always exists globally, see Remark \ref{rem:nls2}. This is quite different comparing with other Camassa--Holm-type equations; particularly, this is different from the case of the SW equation \eqref{mHS-om}, whose RH formalism is very close to that for the SP equation, including the dependence of the RH problem on the parameters $\hat x$ and $t$, see \cite{BSZ11}. The reason is that for all equations mentioned above, the RH formalism relies on the so-called sign condition to be satisfied by the initial data. This condition, on one hand, provides the existence of a global in time solution to the Cauchy problem, and on the other hand, plays a crucial role in introducing the new spatial variable $\hat x$, see \cite{BS06,BS08,BS13,BSZ11}. For instance, this condition reads $-u_{0xx}+1>0$ in the case of the SW equation \eqref{mHS-om}, or $u_0-u_{0xx}+1>0$ in the case of the CH equation \eqref{CH-om}. In the case of the SP equation, the analogous condition would read $1+u_{0x}^2>0$, see the definition of $q$ in \eqref{w-q}, which is obviously automatically satisfied. 
\end{rem}

\subsection{From the RH problem to a solution of the SP equation}

The representation result of Theorem \ref{thm:main} has been actually obtained under assumption of existence of a solution $u(x,t)$ to the Cauchy problem \eqref{spe-ivp}. On the other hand, an important element of the inverse scattering approach to nonlinear equations is the possibility to check directly that a solution of the RH problem with any appropriate $r(k)$ (ensuring the unique solvability of the RH problem) gives rise to a solution of the nonlinear equation in question. The idea consists in direct checking that the solution of the RH problem, properly normalized, satisfies a system of differential equations (w.r.t.\ the \emph{outer parameters} $x$ (or $\hat x$) and $t$), which can be interpreted as the Lax pair for the nonlinear equation. For example, see \cite{FT} for the nonlinear Schr\"odinger equation. For equations of the Camassa--Holm type (and their short wave limits), the procedure is more involved; see, e.g., \cite{BS15} for the case of the short wave limit of the Degasperis--Procesi equation. For the SP equation, the following theorem holds.

\begin{thm}\label{thm:lax}
Let $\{r(k),\,k\in\D{R};\,\{k_j,\gamma_j\}_{1}^N\}$ be a data set such that the RH problem \eqref{jump-y}--\eqref{M-norm} it determines has a unique solution $\hat M(\hat x,t,k)$. Define $f_1$, $f_2$ by \eqref{rep}. Introduce $x(\hat x,t)$ and $\hat u(\hat x,t)$ as in \eqref{xofy}--\eqref{uofx}, and 
\begin{equation} \label{qw-hat}
\hat q(\hat x,t)\coloneqq\frac{1}{\alpha^2(\hat x,t)-\beta^2(\hat x,t)},\qquad\hat w(\hat x,t)\coloneqq\frac{2\alpha(\hat x,t)\beta(\hat x,t)}{\alpha^2(\hat x,t)-\beta^2(\hat x,t)},
\end{equation}
where 
\begin{equation}\label{M0-al}
\begin{pmatrix}
\alpha(\hat x,t)&\beta(\hat x,t)\\-\beta(\hat x,t) & \alpha(\hat x,t)\end{pmatrix}\coloneqq\hat M(\hat x,t,0).
\end{equation}
Then the following equations (between functions of $(\hat x,t)$) hold:
\begin{enumerate}[\rm(a)]
\item 
$x_{\hat x}=\dfrac{1}{\hat q}$;
\item
$\hat u_{\hat x}=\dfrac{\hat w}{\hat q}$;
\item
$\hat q_t=\hat u\hat q\hat w$.
\end{enumerate}
\end{thm}

\begin{rem}
We have already noted (Remark~\ref{rem:sym}) that $\hat M$ satisfies the symmetries \eqref{M-sym}. The specific forms of the l.h.s.\ of \eqref{rep} and \eqref{M0-al} follow from these symmetries, and the functions $f_1(\hat x,t)$, $f_2(\hat x,t)$, $x(\hat x,t)$, $\hat u(\hat x,t)$, $\alpha(\hat x,t)$, and $\beta(\hat x,t)$ are all real-valued for the same reasons. Moreover, $\hat q>1$ because $\alpha^2+\beta^2=\det\hat M(0)=1$.
\end{rem}

\begin{proof}[Proof of Theorem~\ref{thm:lax}]
The proof of (a)-(c) is based on calculations of $\Psi_{\hat x}\Psi^{-1}$ and $\Psi_t\Psi^{-1}$ where
\[
\Psi(\hat x,t,k)\coloneqq\hat M(\hat x,t;k)\eul^{(-\ii k\hat x-\frac{t}{4\ii k})\sigma_3}.
\]
\begin{proof}[Proof of \emph{(a)-(b)}]
We consider $\Psi_{\hat x}\Psi^{-1}$. Starting from the expansion
\[
\hat M(\hat x,t,k)=I+M_1/\ii k+\ord(k^{-2}),\quad k\to\infty
\]
and denoting $W\coloneqq-\croch{M_1,\sigma_3}$, we get
\[
(\Psi_{\hat x}\Psi^{-1})(\hat x,t,k)=-\ii k\sigma_3+W(\hat x,t)+ \ord(k^{-1}),\quad k\to\infty.
\]
Moreover, $(\Psi_{\hat x}\Psi^{-1})(\hat x,t,k)+\ii k\sigma_3$ has neither jumps no singularities and is bounded in $k\in\D{C}$; hence, by Liouville's theorem,
\begin{equation}
(\Psi_{\hat x}\Psi^{-1})(\hat x,t,k)=-\ii k\sigma_3+W(\hat x,t).
\label{psi-y-lax}
\end{equation}
On the other hand, starting from the expansion
\[
\Psi(\hat x,t,k)=G_0(\hat x,t)\left(I-\ii kG_1(\hat x,t)+\ord(k^2)\right)\eul^{(-\ii k\hat x-\frac{t}{4\ii k})\sigma_3},\quad k\to 0,
\]
where, by \eqref{rep} and \eqref{M0-al})
\[
G_0\coloneqq\begin{pmatrix}
\alpha & \beta \\ -\beta & \alpha
\end{pmatrix},\qquad 
G_1\coloneqq\begin{pmatrix}
f_1 &f_2 \\f_2 & -f_1
\end{pmatrix},  
\]
we obtain
\[
\Psi_{\hat x}\Psi^{-1}=G_{0\hat x}G_0^{-1}-\ii kG_0(G_{1\hat x}+ \sigma_3)G_0^{-1}+\ord(k^2),\quad k\to 0.
\]
Comparing this with \eqref{psi-y-lax}, it follows that 
\[
G_{1\hat x}=-\sigma_3+G_0^{-1}\sigma_3G_0=\begin{pmatrix}\alpha^2-\beta^2-1&2\alpha\beta\\2\alpha\beta&\alpha^2-\beta^2-1\end{pmatrix},
\]
which, in terms of $f_1$, $f_2$, $\hat w$, and $\hat q$, reads
\begin{equation}
f_{1\hat x}=\frac{1}{\hat q}-1,\qquad f_{2\hat x}=\frac{\hat w}{\hat q}
\end{equation}
and thus (a) and (b) hold. By \eqref{xofy} and \eqref{uofx} we indeed have $x_{\hat x}=1+f_{1\hat x}$ and $\hat u_{\hat x}=f_{2\hat x}$.
\renewcommand{\qed}{}
\end{proof}
\begin{proof}[Proof of \emph{(c)}]
Now we consider $\Psi_t\Psi^{-1}$. On one hand,
\[
\Psi_t\Psi^{-1}=\ord(k^{-1}),\quad k\to\infty.
\]
On the other hand,
\[
\Psi_t\Psi^{-1}=-\frac{1}{4\ii k}G_0\sigma_3G_0^{-1}+\left\{G_{0t}+\frac{1}{4}G_0\croch{G_1,\sigma_3}\right\}G_0^{-1}, 
\quad k\to 0.
\]
Thus, by Liouville's theorem, 
\[
G_{0t}=-\frac{1}{4}G_0\croch{G_1,\sigma_3}=-\frac{1}{2}\begin{pmatrix}\beta f_2&-\alpha f_2\\\alpha f_2&\beta f_2\end{pmatrix},
\]
which, in terms of  $\hat u$, $\hat w$ and $\hat q$, reads $\hat q_t=\hat u\hat q\hat w$ and thus item (c) of Theorem \ref{thm:lax}
holds. 
\end{proof}
\renewcommand{\qed}{}
\end{proof}

\begin{cor}  \label{cor}
With the same assumptions and notations as in Theorem~\ref{thm:lax} we introduce
\[
u(x,t)\coloneqq\hat u(\hat x(x,t),t),\quad q(x,t)\coloneqq\hat q(\hat x(x,t),t).
\]
Then the three equations \emph{(a)--(c)} from Theorem~\ref{thm:lax} reduce to 
\begin{subequations}\label{cor-ab}
\begin{align}\label{cor-a}
&q_t=\frac{1}{2}(u^2q)_x,\\
\label{cor-b}
&q=\sqrt{1+u_x^2},
\end{align}
\end{subequations} 
which is the SP equation in the conservation law form.
\end{cor}

\begin{proof}
First, it follows from (a) that $\hat x_x(x,t)=q(x,t)$ and from (b) that $\hat u_{\hat x}(\hat x(x,t),t)=\frac{w}{q}(x,t)$, where $w(x,t)\coloneqq\hat w(\hat x(x,t),t)$. Hence, the identity $u_x(x,t)=\hat u_{\hat x}(\hat x(x,t),t)\hat x_x(x,t)$ gives
\begin{equation}\label{wux}
w=u_x.
\end{equation}
Thus, \eqref{cor-b} reads $q=\sqrt{1+w^2}$, or $\hat q=\sqrt{1+\hat w^2}$, which follows from the definitions \eqref{qw-hat} of $\hat q$ and $\hat w$.

In order to derive \eqref{cor-a}, we first notice that (c) can be written in the conservation law form 
\begin{equation}
\left(\frac{1}{\hat q}\right)_t= -\frac{1}{2}\left(\hat u^2\right)_{\hat x}.
\label{cons-hat}
\end{equation}
Indeed, 
\[
\left(\frac{1}{\hat q}\right)_t=-\frac{\hat q_t}{\hat q^2} 
	= -\frac{\hat u \hat w}{\hat q} = -\hat u \hat u_{\hat x}
	= -\frac{1}{2}\left(\hat u^2\right)_{\hat x},
\]
where (c) and then (b) have been used. Now, we calculate $x_t(\hat x, t)$ starting from (a), then using \eqref{cons-hat}:
\[
x_t(\hat x,t) =-\frac{\partial}{\partial t}\left(
\int_{\hat x}^\infty \frac{\dd\xi}{\hat q(\xi,t)}\right)
	= \frac{1}{2}\int_{\hat x}^\infty
	\left(\hat u^2\right)_\xi(\xi,t) \dd\xi 
	= -\frac{1}{2}\hat u^2(\hat x,t).
\]
Substituting this into the identity $\hat q_t=q_xx_t+q_t$ (between functions of $(\hat x,t)$) and using (c) gives $q_t=\hat u\hat q\hat w+\frac{1}{2}q_x\hat u^2$, which reads $q_t=uqw+\frac{1}{2}q_xu^2$ in terms of functions of $(x,t)$. Using \eqref{wux} yields \eqref{cor-a}:
\[
q_t=uqu_x+\frac{1}{2}q_xu^2=\frac{1}{2}(u^2q)_x.\qedhere
\]
\end{proof}

\section{Solitons}\label{sec:4}
	
In the general case, solving a Riemann--Hilbert problem reduces to solving a coupled system of integral equations (generated by the jump condition) and algebraic equations (generated by the residue conditions). In this framework, pure soliton solutions arise in the case where the jump condition is trivial ($J\equiv I$) and thus the solution of the RH problem, being a rational function of the spectral parameter, reduces to solving a system of linear algebraic equations only. The dimension of this system is determined by the number of poles.

In the case of the SP equation, it is natural to distinguish between the solutions associated with pure imaginary zeros of $a(k)$ and those associated with zeros $k_j$ with a nonzero real part.

\subsection{}

First, consider the case, where $a(k)$ has a single, pure imaginary zero at $k=\ii\nu$, $\nu>0$. Then $\hat M$ has two simple poles: one at $k=\ii\nu$ and the second one at $k=-\ii\nu$. From the normalization condition it follows that $\hat M$ has the form
\[
\hat M=\begin{pmatrix}\frac{k-B_{11}}{k-\ii\nu}&\frac{B_{12}}{k+\ii\nu}\\[2mm]
\frac{B_{21}}{k-\ii\nu}&\frac{k-B_{22}}{k+\ii\nu}\end{pmatrix},
\]
where $B_{ij}$ are functions of $\hat x$ and $t$ to be determined from the residue conditions.

Using the symmetry condition $\overline{\hat M}(-\overline{k})=\hat M(k) $ we conclude that $B_{ij}=-\overline{B}_{ij}$, $i,j\in\{1,2\}$ whereas the symmetry $\hat M(-k)=\left(\begin{smallmatrix} 0&1\\-1&0\end{smallmatrix}\right)\hat M(k)\left(\begin{smallmatrix} 0&-1\\1&0\end{smallmatrix}\right)$ implies $B_{11}=-B_{22}$ and $B_{12}=B_{21}$.

Introducing the real-valued functions $b_1$ and $b_2$ by $B_{11}=-B_{22}=\ii b_1$ and $B_{12}=B_{21}=\ii b_2$, $M$ can be written as 
\begin{equation}\label{m-sol}
\hat M=\begin{pmatrix}
\frac{k-\ii b_1}{k-\ii\nu}&\frac{\ii b_2}{k+\ii\nu}\\[2mm] 
\frac{\ii b_2}{k-\ii\nu}&\frac{k+\ii b_1}{k+\ii\nu}\end{pmatrix}.
\end{equation}
Denoting $e_1\coloneqq\eul^{-2\nu\hat{x}-\frac{t}{2\nu}+\log\abs{\gamma}}$ and taking into account that in this case $\gamma\in\D{R}$, the residue conditions \eqref{res-M-hat} take the form 
\begin{equation}\label{res-sol}
\Res_{k=\ii\nu}\hat M^{(1)}=\ii\,\sign(\gamma)e_1\hat M^{(2)}(\ii\nu).
\end{equation}
Taking into account \eqref{m-sol}, these conditions lead to the system of equations for $b_1$ and $b_2$:
\[
\begin{cases} 
\ii\nu-\ii b_1=\ii\,\sign(\gamma)e_1\frac{b_2}{2\nu}&\\
\ii b_2=\ii\,\sign(\gamma)e_1\frac{\nu+b_1}{2\nu}&
\end{cases}
\]
from which $b_1$ and $b_2$ can be determined as follows:
\begin{equation}\label{res1}
b_1=\frac{\nu(4\nu^2-e_1^2)}{4\nu^2+e_1^2},
\qquad  b_2=\frac{4\nu^2\,\sign(\gamma)e_1}{4\nu^2+e_1^2}.
\end{equation}

Thus we have solved the RH problem. In accordance with Theorem \ref{thm:main}, the expansion of $\hat M(\,\cdot\,,\,\cdot\,,k)$ at $k=0$ gives
\[
\hat M(k)=\begin{pmatrix} 
\frac{b_1}{\nu}+\ii k\frac{\nu-b_1}{\nu^2} & \frac{b_2}{\nu}+\ii k\frac{b_2}{\nu^2} \\[2mm] 
-\frac{b_2}{\nu}+\ii k\frac{b_2}{\nu^2} & \frac{b_1}{\nu}+\ii k\frac{b_1-\nu}{\nu^2}\end{pmatrix}+\ord(k^2).
\]
Particularly, 
\[
\hat M(0)=\begin{pmatrix} \frac{b_1}{\nu} & \frac{b_2}{\nu} \\[2mm] \frac{-b_2}{\nu} & \frac{b_1}{\nu}\end{pmatrix}
\]
and thus 
\[
{\hat M(0)}^{-1}\hat M(k)=I-\ii k\begin{pmatrix} \frac{1}{\nu}-\frac{b_1}{b_1^2+b_2^2} & -\frac{b_2}{b_1^2+b_2^2} \\[2mm] -\frac{b_2}{b_1^2+b_2^2} & \frac{b_1}{b_1^2+b_2^2}-\frac{1}{\nu}\end{pmatrix}+\ord(k^2).
\]
Finally, using \eqref{eq:main}, we arrive at 

\begin{thm}[one-soliton]\label{1sol}
One-soliton solutions $u(x,t)$ of the SP equation \eqref{spe} can be expressed, in parametric form, as follows:
\[
u(x,t)=\hat u(\hat x(x,t),t),
\]
where
\begin{equation}\label{sol-1}
\begin{split}
\hat u(\hat x,t)&=-\frac{4\,\sign(\gamma)e_1(\hat x,t)}{4\nu^2 +e_1^2(\hat x,t)},\\
x(\hat x,t)&=\hat x+\frac{2}{\nu}\,\frac{e_1^2(\hat x,t)}{4\nu^2+e_1^2(\hat x,t)},
\end{split}
\end{equation}
with
\[
e_1(\hat x,t)\coloneqq\eul^{-2\nu\hat{x}-\frac{t}{2\nu}+\log\abs{\gamma}}.
\]
Here, $\nu>0$ and $\gamma\in\D{R}$ are the soliton parameters.
\end{thm}

Introducing 
\[
\phi(\hat x,t)=2\nu\left(\hat x+\frac{t}{4\nu^2}-y_0\right),\qquad y_0=\frac{1}{2\nu}\log\frac{|\gamma|}{2\nu}\,,
\]
the soliton formulas \eqref{sol-1} can be written as 
\begin{equation}\label{sol-1-1}
\begin{split}
\hat u(\hat x,t)&=-\frac{\sign\gamma}{\nu}\frac{1}{\cosh\phi(\hat x,t)},\\
x(\hat x,t)&=\hat x+\frac{1}{\nu}\left(1-\tanh\phi(\hat x,t)\right).
\end{split}
\end{equation}

We notice that the one-soliton solution described above is always a multivalued function having the form of a loop. Indeed, $\frac{\partial x}{\partial\hat x}=2\tanh^2 \phi-1$, which changes sign twice as $\phi$ is varying from $-\infty$ to $+\infty$.

For $\nu=\frac{1}{2}$ and $\gamma=-1$, \eqref{sol-1-1} reads
\begin{equation}\label{22}
\begin{split}
\hat u(\hat x,t)&=\frac{2}{\cosh (\hat x+t)},\\
x(\hat x,t)&=\hat x-2\tanh(\hat x+t)+2,
\end{split}
\end{equation}
and thus we retrieve the formulas for the soliton solution presented in \cite{ss06} (comparing with \cite{ss06}, the additional constant in \eqref{22} provides that $x-\hat x\to 0$ as $\hat x\to +\infty$, cf.~\eqref{x-hat}), where they were obtained using the connection between the short pulse equation and the sine-Gordon equation.

\subsection{}

Now consider the case, where $a(k)$ has a pair of zeros: $a(k_0)=0=a(-\bar k_0)$ with
\[
k_0=\mu+\ii\nu,\quad \mu>0,\,\nu>0.
\]
In this case, the  symmetries \eqref{M-sym} and the normalization condition lead to 
\begin{equation}\label{m-sol-2}
\hat M=\frac{1}{2}
\begin{pmatrix} 
\frac{k-b_1}{k-k_0}+\frac{k+\bar b_1}{k+\bar k_0}
&\frac{b_2}{k-\bar k_0}-\frac{\bar b_2}{k+ k_0} \\[2mm] 
-\frac{\bar b_2}{k- k_0}+\frac{b_2}{k+ \bar k_0}
&\frac{k-\bar b_1}{k-\bar k_0}+\frac{k+ b_1}{k+ k_0}\end{pmatrix}
\end{equation}
(cf.~\eqref{m-sol}), where $b_1(\hat x,t)$ and $b_2(\hat x,t)$ can be found solving the system of linear equations resulting from the residue conditions \eqref{res-M-hat} at $k=k_0$:
\begin{equation}\label{res-M-hat-br}
\Res_{k=k_0}\hat M^{(1)}(\hat x,t,k)=\ii\gamma_0\eul^{2\ii \bigl(k_0\hat{x}-\frac{t}{4k_0}\bigr)}\hat M^{(2)}(\hat x,t,k_0) 
\end{equation}
(the other residue conditions at $-\bar k_0$, $-k_0$, and $\bar k_0$ then follow from the symmetry condition).

Introducing $\gamma=\abs{\gamma}\eul^{\ii\arg\gamma}$ and writing $\gamma_0\eul^{2\ii\bigl(k_0\hat{x}-\frac{t}{4k_0}\bigr)}$ in \eqref{res-M-hat-br} as
\[
\gamma_0\eul^{2\ii\bigl(k_0\hat{x}-\frac{t}{4k_0}\bigr)}=\eul^{-\frac{\nu}{\abs{k_0}}\left(2\abs{k_0}\hat x+\frac{t}{2\abs{k_0}}-\frac{\abs{k_0}\log\abs{\gamma}}{\nu}\right)}\eul^{\frac{\ii\mu}{\abs{k_0}}\left(2\abs{k_0}\hat x-\frac{t}{2\abs{k_0}}+\frac{\abs{k_0}\arg\gamma}{\mu}
\right)}
\]
suggest introducing 
\begin{align*}
\phi&=\frac{\nu}{\sqrt{\nu^2+\mu^2}}\left(2\abs{k_0}\hat x+\frac{t}{2\abs{k_0}}-\frac{\abs{k_0}\log\abs{\gamma}}{\nu}\right),\\
\psi&=\frac{\mu}{\sqrt{\nu^2+\mu^2}}\left(2\abs{k_0}\hat x-\frac{t}{2\abs{k_0}}+\frac{\abs{k_0}\arg\gamma}{\mu}\right),
\end{align*}
in terms of which the solution of the SP equation (after solving \eqref{m-sol-2}, \eqref{res-M-hat-br} for $b_1$ and $b_2$) is given by (see \cite{LPS09,ss06})
\begin{subequations} \label{sol-2}
\begin{align}\label{sol-2-u}
&\hat u(\hat x,t)=\frac{2\mu\nu}{\nu^2+\mu^2}\frac{\nu\sin\psi\sinh\phi+\mu\cos\psi\cosh\phi}{\nu^2\sin^2\psi+\mu^2\cosh^2\phi},\\ 
\label{sol-2-x}
&x(\hat x,t)=\hat x+\frac{\mu\nu}{\nu^2+\mu^2}\left(\frac{\nu\sin(2\psi)-\mu\sinh(2\phi)}{\nu^2\sin^2\psi+\mu^2\cosh^2\phi}+\frac{2}{\mu}\right).
\end{align}
\end{subequations}
Observing that (see \cite{LPS09})
\[
\frac{\partial x}{\partial\hat x}=1-\frac{8\mu^2\nu^2\sin^2\psi\cosh^2\phi}{(\nu^2+\mu^2)(\nu^2\sin^2\psi+\mu^2\cosh^2\phi)^2}=\cos\left(4\arctan\frac{\nu\sin\psi}{\mu\cosh\phi}\right),
\]
we see that if $\left|\frac{\nu}{\mu}\right|<\tan\frac{\pi}{8}$, then $\frac{\partial x}{\partial\hat x}>0$ for all $x$, and thus \eqref{sol-2} represents a smooth solution -- the breather. On the other hand, if $\left|\frac{\nu}{\mu}\right|>\tan \frac{\pi}{8}$, then $\frac{\partial x}{\partial \hat x}$ is not sign-definite, and thus \eqref{sol-2} represents a multivalued solution.

\section {Long-time asymptotics}   \label{sec:5}

A major advantage of the representation of the solution $u$ to the Cauchy problem for a nonlinear integrable equation in terms of the solution of an associated Riemann--Hilbert problem is that it can be efficiently used for studying \emph{in details} the long-time behavior of the former problem via the long-time analysis of the latter, applying the nonlinear steepest descent method introduced by Deift and Zhou \cite{DZ93}. For Camassa--Holm-type equations, this approach has been presented in \cite{BKST09,BS-D,BS13,BS15}. A key feature of this method is the deformation of the original RH problem according to the ``signature table'' for the phase function $\theta$ in the jump matrix $\hat J$ written in the form (cf.~\eqref{J-J0}, \eqref{Q-hat})
\begin{equation}\label{jump-theta}
\hat J(\hat x,t;k)=\eul^{-\ii t\theta(\hat\zeta, k)\sigma_3}J_0(k)\eul^{\ii t\theta(\hat\zeta,k)\sigma_3},
\end{equation}
where 
\begin{align}  \label{theta}
&\theta(\hat\zeta,k)=\hat\zeta k-\frac{1}{4k}\,,\\
\label{hatzeta}
&\hat\zeta\coloneqq\frac{\hat x}{t}\,.
\end{align}
The signature table is the distribution of signs of $\Im\theta(\hat\zeta,k)$ in the $k$-plane, depending on the values of $\hat\zeta$. In the case of the SP equation,
\[
\Im\theta(\hat\zeta,k)=\Im k\cdot\Bigl(\hat\zeta+\frac{1}{4|k|^2}\Bigr).
\]

Now we notice that $\hat J(\hat x,t;k)$ in the present case looks very similar to the case of the SW equation, see \cite{BSZ11}, including the matrix structure of $\hat J_0(k)$ and the form of $\theta(\hat\zeta,k)$. Namely, the latter in the case of the SW equation has the form $\theta(\hat\zeta,k)=\hat\zeta k-\frac{1}{2k}$ and thus the distribution of signs of $\Im\theta(\hat\zeta,k)$ is the same modulo the scaling factor $\frac{1}{2}$. As for the structure of $\hat J_0(k)$, the difference with the case of the SW equation is that $\bar r$ is to be replaced by $-\bar r$ while keeping $r$ the same. A direct consequence of this is that the long time analysis in the case of the SP equation repeats the steps made in the case of the SW equation. As for the differences, we notice the following.
\begin{enumerate}[1)]
\item 
The basic difference is that in the case of the SP equation, the RH problem involves, in general, residue conditions (absent in the case of the SW equation). Here there is a complete analogy with the NLS equation \cite{FT}, where the defocusing NLS equation corresponds to the SW equation whereas the focusing NLS equation corresponds to the SP equation. Consequently, if $a(k)$ has zeros, then the solitons associated with the residue conditions dominate the long time behavior of the solution of the Cauchy problem. 
\item
In the solitonless case ($a(k)\neq 0 $ for all $k$ with $\Im k\geq 0$), the main asymptotic term is expressed in terms of the solution of the model RH problem, which is different from that in the SW case. More precisely, the model problem for the SW equation is exactly as in \cite[Appendix B]{L-lr} whereas in the case of the SP equation, the jump matrix of the model problem is as in \cite[(B.1)]{L-lr}, with $\bar q$ replaced by $-\bar q$ (keeping $q$ the same). Accordingly, the large-$z$ expansion of the solution of the model problem, which is 
\[
m^X(q,\hat k)=I+\frac{\ii}{\hat k}
\begin{pmatrix}
0 & -\beta^X(q) \\ 
\overline{\beta^X(q)} & 0
\end{pmatrix}+\ord(\hat k^{-2}),\qquad\hat k\to\infty
\]
from \cite[(B.2)]{L-lr}, is to be replaced by 
\begin{equation}\label{mod-as}
m^X(q,\hat k)=I+\frac{\ii}{\hat k}
\begin{pmatrix}
0 & -\beta^X(q)\\ -\overline{\beta^X(q)} & 0
\end{pmatrix}+\ord(\hat k^{-2}),\qquad\hat k\to\infty
\end{equation}
in the case of the SP equation, with 
\begin{equation}\label{be-x}
\beta^X(q)=\sqrt{-h(q)}\,\eul^{\ii\bigl(\frac{\pi}{4}-\arg q+\arg \Gamma(\ii\nu(q))\bigr)}
\end{equation}
(notice that $h(q)<0$ in this case), where $\Gamma$ is the Euler Gamma function.
\end{enumerate}

Now we are going to give a sketch of the asymptotic analysis in the solitonless case and present an exact asymptotic result. Since the distribution of signs of $\Im\theta(\hat\zeta,k)$ is as in the case of the SW equation, the long time behavior of $u$ is qualitatively different in the same two ranges of values of 
\begin{equation}  \label{zeta}
\zeta\coloneqq\frac{x}{t}\,.
\end{equation}

\subsection{Range $\BS{\zeta>\varepsilon}$}

In this case the set $\{k\mid\Im\theta(\hat\zeta,k)=0\}$ coincides with the real axis $\Im k=0$ and $\pm\Im\theta>0$ for $\pm\Im k>0$. This suggests the use of the following  factorization of the jump matrix for all $k\in\D{R}$:
\begin{equation}\label{factor1}
\hat J=\begin{pmatrix}
1 & \bar r(\bar k)\eul^{-2\ii t \theta} \\ 
0 & 1 
\end{pmatrix}
 \begin{pmatrix}
1 & 0 \\
  r(k)\eul^{2\ii t \theta} & 1 
\end{pmatrix}.
\end{equation}
Indeed, the triangular factors in \eqref{factor1} can be absorbed into a new RH problem for $\hat M^{(1)}(\hat x,t,k)$ in  the same way as in the case  of the SW equation \cite{BSZ11}:
\[
\hat M^{(1)}=\begin{cases}
\hat M\begin{pmatrix}
1 & 0 \\
 - r(k)\eul^{2\ii t \theta} & 1
\end{pmatrix}, & 0<\Im k<\varepsilon, \\
\hat M\begin{pmatrix}
1 & \bar r(\bar k)\eul^{-2\ii t \theta} \\ 
0 & 1 
\end{pmatrix}, & -\varepsilon<\Im k <0, \\
\hat M, & \text{otherwise}.
\end{cases}
\]
This reduces the RH problem to a RH problem with a jump matrix that decays exponentially (in $t$) to the identity matrix. Since this RH problem is holomorphic (there is no residue condition), its solution decays fast to $I$ and consequently $\hat u(\hat x,t)$ decays fast to $0$ while $\hat x$ approaches fast $x$, and thus the domains $\hat\zeta>\varepsilon$ and $\zeta>\varepsilon$ coincide asymptotically.

\subsection{Range $\BS{\zeta<-\varepsilon}$}

Similarly to \cite{BSZ11}, in a domain of the form $\hat\zeta<-\varepsilon$ for any $\varepsilon>0$, the signature table dictates the use of two factorizations. Let $\pm\hat\kappa$ be the points where the distribution of signs is changing:
\begin{equation}\label{kap}
\hat\kappa=\frac{1}{2\sqrt{\abs{\hat\zeta}}}\,.
\end{equation}
\begin{enumerate}[i)]
\item
For $k\in (-\hat\kappa,\hat\kappa)$ we consider again the factorization \eqref{factor1}
\[
\hat J=\begin{pmatrix}
1 & \bar r(\bar k)\eul^{-2\ii t\theta} \\ 
0 & 1 
\end{pmatrix}
 \begin{pmatrix}
1 & 0 \\
  r(k)\eul^{2\ii t\theta} & 1 
\end{pmatrix}.
\]
\item
For $k\in(-\infty,-\hat\kappa)\cup(\hat\kappa,\infty)$ we consider a factorization with triangular factors in reverse order:
\begin{equation}\label{factor2}
\hat J=\begin{pmatrix}
1 & 0 \\
\frac{r(k)}{1-\abs{r(k)}^2}\eul^{2\ii t\theta} & 1
\end{pmatrix}
 \begin{pmatrix}
1-\abs{r(k)}^2 & 0 \\
0 & \frac{1}{1-\abs{r(k)}^2}
\end{pmatrix}
\begin{pmatrix}
1 & \frac{\bar r(\bar k)}{1-\abs{r(k)}^2}\eul^{-2\ii t\theta} \\
0 & 1 
\end{pmatrix}
\end{equation}
\end{enumerate}
Similarly to the previous case, an appropriate sequence of deformations of the RH problem is the same as in the case of the SW equation, so we will follow it giving details mainly for items specific to the considered equation.

The deformations involve the removal of the diagonal factor in \eqref{factor2} and the consequent absorption of the triangular factors, leading, after an appropriate rescaling, to a model RH problem on a contour consisting of two crosses centered at $k=\pm\hat\kappa$, see \cite{BKST09,BS-D}, which finally leads to the asymptotics in the form of modulated decaying (of the order $\ord(t^{-1/2})$) oscillations. The diagonal term is removed introducing $\hat M^{(1)}=\hat M\delta^{-\sigma_3}$, where
\begin{equation}\label{delta}
\delta(k;\hat\zeta)=\exp\left\{\frac{1}{2\pi\ii}\left(\int_{-\infty}^{-\hat\kappa}+\int_{\hat\kappa}^{\infty}\right)\log(1+\abs{r(s)}^2)\frac{\dd s}{s-k}\right\}.
\end{equation}
solves the scalar RH problem whose jump condition is
\[
\delta_+=\delta_-(1+\abs{r(k)}^2)
\]
across the contour $(-\infty,-\hat\kappa)\cup(\hat\kappa,\infty)$. 

The triangular factors are absorbed into the RH problem for $\hat M^{(2)}$:
\begin{equation}\label{mu2}
\hat M^{(2)}=\begin{cases}
\hat M^{(1)}\begin{pmatrix}
1 & 0 \\
-r\delta^{-2}\eul^{2\ii t\theta} & 1
\end{pmatrix},&\Im k>0,\ k\text{ near }(-\hat\kappa,\hat\kappa),\\
\hat M^{(1)}\begin{pmatrix}
1 & -\frac{\bar r}{1-|r|^2}\delta_+^2 \eul^{-2\ii t\theta} \\
0 & 1\end{pmatrix},&\Im k>0,\ k\text{ near }\D{R}\setminus[-\hat\kappa,\hat\kappa],\\
\hat M^{(1)}\begin{pmatrix}
1 & \bar r\delta^{2}\eul^{-2\ii t\theta} \\
0 & 1
\end{pmatrix},&\Im k<0,\ k\text{ near }(-\hat\kappa,\hat\kappa),\\
\hat M^{(1)}\begin{pmatrix}
1 & 0 \\
\frac{r}{1-|r|^2}\delta_-^{-2}\eul^{2\ii t\theta}  & 1
\end{pmatrix},&\Im k<0,\ k\text{ near }\D{R}\setminus[-\hat\kappa,\hat\kappa].
\end{cases}
\end{equation}

Now, in order to reduce the RH problem for $\hat M^{(2)}$, as $t\to\infty$, to a model problem whose solution can be given explicitly in terms of parabolic cylinder functions, see \cite{DZ93,BS-D,L-lr,L16}, the leading term of the factor $\delta(k)\eul^{-\ii t \theta(k)}$ as $k\to\pm\hat\kappa$ is to be evaluated. One has
\begin{equation}\label{delta1}
\delta(k)=\left(\frac{\hat\kappa-k}{\hat\kappa+k}\right)^{-\ii h}\eul^{\chi(k)}
\end{equation}
with
\begin{align}\label{h}
&h\equiv h(\hat\kappa)=-\frac{1}{2\pi}\log\bigl(1+\abs{r(\hat\kappa)}^2\bigr),\\
\label{chi}
&\chi(k)=\frac{1}{2\pi\ii}\int_{\D{R}\setminus\croch{-\hat\kappa,\hat\kappa}}
\frac{\log(1+\abs{r(s)}^2)}{\log(1+\abs{r(\hat\kappa)}^2)} \frac{\dd s}{s-k}.
\end{align}
As $k\to -\hat\kappa$,
\begin{equation}\label{theta1}
\theta(k)=\frac{1}{2\hat\kappa} +\frac{1}{4\hat\kappa^3}(k+\hat\kappa)^2+\ord((k+\hat\kappa)^3).
\end{equation}
Therefore, introducing the scaled spectral variable $\hat k$ by
\begin{equation}\label{k-hat}
k+\hat\kappa=\frac{\hat k}{\sqrt{\hat\kappa^{-3}t}},
\end{equation}
the factor $\delta^2(k)\eul^{-2\ii t\theta(k)}$ can be approximated as 
\begin{equation}\label{ap}
\delta^2(k)\eul^{-2\ii t \theta(k)} \approx \tilde\delta^2 \hat k^{2\ii h}\eul^{-\ii\hat k ^2/2},
\end{equation}
where 
\begin{equation}\label{delta-t}
\tilde\delta^2=\left(\frac{4t}{\hat\kappa}\right)^{-\ii h}\eul^{-\frac{\ii t}{\hat\kappa}}\eul^{-2\chi(\hat\kappa)}.
\end{equation}
Similarly for $k$ near $\hat\kappa$.

Following \cite{L-lr}, the solution of the RH problem for $\hat M^{(2)}$ formulated on two crosses centered at $k=\pm\hat\kappa$ with the jump matrix $\hat J^{(2)}=(\hat M_-^{(2)})^{-1}\hat M_+^{(2)}$ that follows from \eqref{mu2}, can be approximated, for large $t$, in terms of the solution $m^X$ of the model problem formulated in the $\hat k$-plane on a cross centered at $\hat k=0$ and evaluated  for large $\hat k$. In our case, the evaluation of the model problem is given by \eqref{mod-as}, \eqref{be-x} with $q=r(-\hat k)$ and $h$ as in \eqref{h}. Taking into account that we are interested in the expansion of $\hat M^{(2)}(k)$ as $k\to 0$, a reasoning similar to \cite[Eqs~(2.36)--(2-39)]{L-lr} leads to an approximation for $\hat M^{(2)}(k)$ in terms of $m^X$:
\begin{align}\label{M-hat-exp}
\hat M^{(2)}(k)&=I+\frac{1}{\pi}\Im\int_{|k+\hat\kappa|=\rho}(m_0^{-1}(s)-I) \frac{\dd s}{s}\notag\\
&\quad+\frac{k}{\ii\pi}\Re\int_{|k+\hat \kappa|=\rho}(m_0^{-1}(s)-I)\frac{\dd s}{s^2}+\ord(k^2t^{-1/2-\varepsilon})
\end{align}
with some $\rho>0$ and $\varepsilon>0$, where
\[
m_0(k)=\tilde\delta^{\sigma_3}m^X\left(\hat\zeta,\sqrt{\frac{t}{\hat \kappa^3}}(k+\hat\kappa)\right)\tilde\delta^{-\sigma_3}.
\]
In view of \eqref{mod-as} and \eqref{k-hat}, and for large $t$,
\begin{equation}\label{m0-exp}
m_0(k)=I+\frac{\ii\sqrt{\hat\kappa^3}}{\sqrt{t}(k+\hat\kappa)}
\begin{pmatrix}
0 & -\beta^X\tilde \delta^2 \\ 
-\bar \beta^X\tilde \delta^{-2} & 0
\end{pmatrix}+\ord(t^{-1/2-\varepsilon}).
\end{equation}

Now recall that $\hat M=\hat M^{(1)}\delta^{\sigma_3}$ and that $\hat M^{(1)}$ is related to $\hat M^{(2)}$ by \eqref{mu2}, where $r(k)=\ord(k^2)$. Evaluating $\delta(k;\hat\zeta)$ as $k\to 0$ gives
\[
\delta(k;\hat\zeta)=1-\ii k Q + \ord(k^2),
\]
where 
\[
Q=\frac{1}{\pi}\int_{\hat\kappa}^\infty \frac{\log(1+\abs{r(s)}^2)}{s^2}\,\dd s.
\]
Taking this into account and substituting \eqref{m0-exp} into \eqref{M-hat-exp}, one obtains
\[
\hat M(k)=I + \frac{c_1}{\sqrt{t}}\begin{pmatrix}
0 & 1 \\ -1 & 0
\end{pmatrix}+\ii k\left(-Q \sigma_3 +\frac{Qc_1+c_2}{\sqrt{t}}
\begin{pmatrix}
0 & 1 \\ 1 & 0
\end{pmatrix}
\right)+\ord(k^2 t^{-1/2-\varepsilon}),
\]
where 
\[
c_1=2\sqrt{\hat\kappa}\Im \{\beta^X\tilde\delta^2\},
\qquad 
c_2=\frac{2}{\sqrt{\hat\kappa}}\Re \{\beta^X\tilde\delta^2\}.
\]
Finally, using \eqref{xofy}--\eqref{rep} one arrives at the asymptotic formulas
\[
\hat u=\frac{c_2}{\sqrt{t}}(1+\osmall(1)),\qquad 
x-\hat x=Q(1+\osmall(1)), \qquad t\to\infty,
\]
which imply the asymptotics for $u$ in the original variables:
\begin{equation}\label{u-ass}
u(x,t)=2\sqrt{\frac{|h(\varkappa)|}{\varkappa\,t}}\cos\left\{\frac{t}{\varkappa}+h(\varkappa)\log t+\phi_0(\varkappa)\right\},
\end{equation}
where
\begin{equation}\label{varkappa}
\varkappa=\frac{1}{2}\sqrt\frac{t}{\abs{x}}
\end{equation}
and
\begin{align}\label{phi0}
\phi_0(\varkappa)&=-\frac{\pi}{4}-\arg(r(\varkappa))-\arg\Gamma(\ii h(\varkappa))+\frac{1}{\pi}\int_{\D{R}\setminus\croch{-\varkappa,\varkappa}}\log|k-s|\,\dd\log(1+|r(s)|^2)\notag\\
&\quad+\frac{2\varkappa}{\pi}\int_{\varkappa}^\infty\frac{\log(1+|r(s)|^2)}{s^2}\dd s+h(\varkappa)\log\frac{4}{\varkappa}\,.
\end{align}

\begin{thm}[solitonless asymptotics]   \label{thm:asymptotics}
Let $u(x,t)$ be the solution of the Cauchy problem \eqref{spe-ivp}. Assume that the spectral function $a(k)$ constructed from $u_0(x)$ has no zeros in the upper half-plane. Then the behavior of $u$ as $t\to\infty$ is described as follows. Let $\varepsilon$ be any small positive number.
\begin{enumerate}[\rm(i)]
\item
In the domain $\zeta\equiv x/t>\varepsilon$, $u(x,t)$ tends to $0$ with fast decay.
\item
In the domain $\zeta\equiv x/t<-\varepsilon$, $u(x,t)$ exhibits decaying (of the order $\ord(t^{-1/2})$) modulated oscillations given by \eqref{u-ass}, where $h(\varkappa)$ and $\phi_0(\varkappa)$ are functions of $\varkappa=1/(2\sqrt{\abs{\zeta}})$ given in terms of the associated reflection coefficient $r(k)$; in particular, by \eqref{h} and \eqref{phi0} 
\begin{equation}\label{h-vark}
h(\varkappa)=-\frac{1}{2\pi}\log\bigl(1+\abs{r(\varkappa)}^2\bigr).
\end{equation}
\end{enumerate}
\end{thm}

For completeness, we present also the asymptotics in the soliton case.

\begin{thm}[soliton asymptotics]   \label{thm:asymptotics-sol}
Assume that $a(k)$ has $N=2n+m$ simple zeros
\[
\{k_j\}_1^n\cup \{-\bar k_j\}_1^n\cup \{i\nu_j\}_1^m,
\]
where $\mu_j\coloneqq\Re k_j>0$, $\nu_j\coloneqq\Im k_j>0$. Assume also that if $\mu_j^2+\nu_j^2\neq\mu_l^2+\nu_l^2$ if $j\neq l$. Then:
\begin{enumerate}[\rm(i)]
\item 
For $\varepsilon>0$ sufficiently small, the asymptotics of $u$ in each sector $\left|\frac{x}{t}+\frac{1}{4(\mu_j^2+\nu_j^2)}\right|<\varepsilon$ is given by
\[
u(x,t)=u_j(x,t)+\ord(t^{-1/2}),
\]
where $u_j$ is given, parametrically, as follows:
\begin{enumerate}[\rm a)]
\item 
If $\mu_j=0$, then $u_j$ is given by \eqref{sol-1-1} with $\nu$ replaced by $\nu_j$ and $\phi$ replaced by 
\[
\phi_j = 2\nu_j \hat x +\frac{t}{2\nu_j} + \phi_j^0.
\]
If $\mu_j\neq 0$, then $u_j$ is given by \eqref{sol-2}, with $\mu$, $\nu$, $\phi$, and $\psi$ replaced respectively by 
$\mu_j$, $\nu_j$,
\[
\phi_j = 2\nu_j \hat x +\frac{\nu_j t}{2 (\mu_j^2+\nu_j^2)} + \phi_j^0,\qquad
\psi_j = 2\mu_j \hat x -\frac{\mu_j t}{2 (\mu_j^2+\nu_j^2)} + \psi_j^0.
\]
\end{enumerate}
Here $\phi_j^0 $ and $\psi_j^0 $ are constants determined by the 
scattering data $\{r(k),k\in\D{R}; \{k_j,\gamma_j\}_{1}^N\}$.
\item 
Outside these sectors, $u(x,t)=\ord(t^{-1/2})$.
\end{enumerate}
\end{thm}

\begin{rem}
The asymptotic results presented above imply the following.
\begin{enumerate}[a)]
\item 
In the solitonless case as well as in the case when all the zeros $k_j$ of $a(k)$ are located outside the sector $\left|\frac{\Im k_j}{\Re k_j}\right|\geq\tan\frac{\pi}{8}$, there exists $T>0$ such that for all $t>T$, the solution of the Cauchy problem \eqref{spe-ivp} is a smooth classical solution (possibly after passing through wave breakings).
\item
A sufficient condition for wave breaking: If $a(k)$ has a zero $k^*$ in the sector $\left|\frac{\Im k^*}{\Re k^*}\right|\geq\tan\frac{\pi}{8}$, then wave breaking occurs at a certain finite time. 

Notice that another sufficient condition for finite time wave breaking has been obtained in \cite{LPS09} using the method of characteristics and conserved quantities.
\end{enumerate}
\end{rem}

\begin{rem}
The asymptotic formula presented in Theorem \ref{thm:asymptotics}
improves the asymptotics obtained in \cite{HN15}  and \cite{N14}
by using different methods (not relying on the integrability of the SP equation).
\end{rem}



\begin{thebibliography}{10}
\bibitem{BC}
R.~Beals and R.~R. Coifman.
\newblock 
Scattering and inverse scattering for first order systems.
\newblock 
{\em Comm. Pure Appl. Math.}, 37(1):39--90, 1984.

\bibitem{BKST09}
A. Boutet~de Monvel, A. Kostenko, D. Shepelsky, and G. Teschl.
\newblock 
Long-time asymptotics for the Camassa--Holm equation.
\newblock 
{\em SIAM J. Math. Anal.}, 41(4):1559--1588, 2009.

\bibitem{BS06}
A. Boutet~de Monvel and D. Shepelsky.
\newblock 
Riemann-Hilbert approach for the Camassa--Holm equation on the line.
\newblock
{\em C. R. Math. Acad. Sci. Paris}, 343(10):627--632, 2006.

\bibitem{BS-D}
A. Boutet~de Monvel and D. Shepelsky.
\newblock 
Long-time asymptotics of the Camassa--Holm equation on the line.
\newblock 
In {\em Integrable systems and random matrices}, volume 458 of {\em Contemp. Math.}, pages 99--116. Amer. Math. Soc., Providence, RI, 2008.

\bibitem{BS08}
A. Boutet~de Monvel and D. Shepelsky.
\newblock 
Riemann--Hilbert problem in the inverse scattering for the Camassa--Holm equation on the line.
\newblock 
In {\em Probability, geometry and integrable systems}, volume~55 of {\em Math. Sci. Res. Inst. Publ.}, pages 53--75. Cambridge Univ. Press, Cambridge, 2008.

\bibitem{BS13}
A. Boutet~de Monvel and D. Shepelsky.
\newblock 
A Riemann--Hilbert approach for the Degasperis--Procesi equation.
\newblock 
{\em Nonlinearity}, 26:2081--2107, 2013.

\bibitem{BS15}
A. Boutet~de Monvel and D. Shepelsky.
\newblock 
The Ostrovsky--Vakhnenko equation by a Riemann--Hilbert approach.
\newblock 
{\em J. Phys. A: Math. Theor.}, 48:035204, 2015.

\bibitem{BSZ11}
A. Boutet~de Monvel, D. Shepelsky, and L. Zielinski.
\newblock 
The short-wave model for the Camassa--Holm equation: a Riemann--Hilbert approach.
\newblock 
{\em Inverse Problems}, 27:105006, 2011.

\bibitem{B05}
J. C. Brunelli.
\newblock 
The short pulse hierarchy.
\newblock 
{\em J. Math. Phys.}, 46:123507, 2005.

\bibitem{B06}
J. C. Brunelli.
\newblock 
The bi-Hamiltonian structure of the short pulse equation.
\newblock 
{\em Phys. Lett. A}, 353(6):475--478, 2006.

\bibitem{CH93}
R. Camassa and D.~D. Holm.
\newblock 
An integrable shallow water equation with peaked solitons.
\newblock 
{\em Phys. Rev. Lett.}, 71(11):1661--1664, 1993.

\bibitem{sw05}
Y. Chung, C. K. R. T. Jones, T. Sch\"afer, and C. E. Wayne.
\newblock 
Ultra-short pulses in linear and nonlinear media
\newblock 
{\em Nonlinearity}, 18(3):1351--1374, 2005.

\bibitem{CR15}
G. M. Coclite and L. di Ruvo, 
\newblock 
Well-posedness results for the short pulse equation.
\newblock 
{\em Z. Angew. Math. Phys.}, 66:1529--1557, 2015.

\bibitem{CR15-1}
G. M. Coclite and L. di Ruvo, 
\newblock 
Wellposedness of bounded solutions of the non-homogeneous initial boundary for the short pulse equation.
\newblock 
{\em Boll. Unione Mat. Ital.}, 8(1):31--44, 2015.

\bibitem{DZ93}
P.~Deift and X.~Zhou.
\newblock 
A steepest descent method for oscillatory Riemann--Hilbert problems. Asymptotics for the MKdV equation.
\newblock 
{\em Ann. of Math. (2)}, 137(2):295--368, 1993.

\bibitem{FT}
L.~D.~Faddeev and L.~A.~Takhtajan,  
\newblock 
{\em Hamiltonian methods in the theory of solitons.} 
\newblock 
Springer Series in Soviet Mathematics. Springer-Verlag, Berlin, 1987.

\bibitem{F96}
B. Fuchssteiner.
\newblock 
Some tricks from the symmetry-toolbox for nonlinear equations: generalizations of the Camassa--Holm equation.
\newblock 
{\em Phys. D}, 95:229--243, 1996.

\bibitem{FF81}
B. Fuchssteiner and A.~S. Fokas.
\newblock 
Symplectic structures, their B\"acklund transformations and hereditary symmetries.
\newblock 
{\em Phys. D}, 4:47--66, 1981.

\bibitem{G15}
G. Gambino, U. Tanriver, P. Guha, A. Choudhury, and S. Choudhury.
\newblock  
Regular and singular pulse and front solutions and possible isochronous behavior in the short-pulse equation: phase-plane, multi-infinite series and variational approaches.
\newblock
{\em Commun. Nonlinear Sci. Numer. Simul.}, 20(2):375--388, 2015.

\bibitem{HN15} 
N. Hayashi and P. Naumkin.
\newblock
Large time asymptotics for the reduced Ostrovsky equation.
\newblock
{\em Commun. Math. Phys.} 335:713--738, 2015. 

\bibitem{L16}
J. Lenells.
\newblock
The nonlinear steepest descent method: asymptotics for initial-boundary value problems.
\newblock
{\em SIAM J. Math. Anal.} 48(3):2076--2118, 2016. 

\bibitem{L-lr}
J. Lenells.
\newblock
The nonlinear steepest descent method for Riemann--Hilbert problems of low regularity.
\newblock
arXiv:1501.05329.

\bibitem{LPS09}
Y. Liu, D. Pelinovsky, and A. Sakovich.
\newblock
Wave breaking in the short-pulse equation.
\newblock
{\em Dynamics of PDE}, 6:291--310, 2009.

\bibitem{M07} 
Y. Matsuno.
\newblock
Multiloop soliton and multibreather solutions of the short pulse model equation.
\newblock
{\em J. Phys. Soc. Japan}, 76(8):084003, 6 pp., 2007.

\bibitem{M08} 
Y. Matsuno.
\newblock
Periodic solutions of the short pulse model equation.
\newblock
{\em J. Math. Phys.}, 49(7):073508, 18 pp., 2008.

\bibitem{N14} 
T. Niizato.
\newblock
Asymptotic behavior of solutions to the short pulse equation with critical nonlinearity.
\newblock
{\em Nonlinear Analysis}, 111:15--32, 2014.

\bibitem{OR96}
P. J. Olver and P. Rosenau.
\newblock
Tri-Hamiltonian duality between solitons and solitary-wave solutions having compact support.
\newblock
{\em Phys. Rev. E}, 53:1900--1906, 1996.

\bibitem{PS10}
D. Pelinovsky and A. Sakovich.
\newblock
Global well-posedness of the short-pulse and sine-Gordon equations in energy space.
\newblock, 
{\em Comm. Partial Differential Equations}, 35(4):613--629, 2010.

\bibitem{Q06}
Z. Qiao.
\newblock
A new integrable equation with cuspons and W/M-shape-peaks solitons.
\newblock 
{\em J. Math. Phys.}, 47:112701, 2006.

\bibitem{ss05}
A. Sakovich and S. Sakovich.
\newblock
The short pulse equation is integrable.
\newblock
{\em J. Phys. Soc. Japan}, 74:239--241, 2005.
	
\bibitem{ss06}
A. Sakovich and S. Sakovich.
\newblock
Solitary wave solutions of the short pulse equation.
\newblock, 
{\em J. Phys. A: Math. Gen.}, 39:L361--L367, 2006.

\bibitem{sw04}
T. Sch\"afer and C. E. Wayne.
 \newblock
Propagation of ultra-short optical pulse in nonlinear media.
\newblock
{\em Physica D}, 196:90-105, 2004.
\end{thebibliography}
\end{document}